\documentclass[conference]{IEEEtran}
\makeatletter
\def\ps@headings{%
\def\@oddhead{\mbox{}\scriptsize\rightmark \hfil \thepage}%
\def\@evenhead{\scriptsize\thepage \hfil \leftmark\mbox{}}%
\def\@oddfoot{}%
\def\@evenfoot{}}
\makeatother
\usepackage{graphicx,graphics,subfigure,wrapfig}
\usepackage{epsfig}
\usepackage{psfrag}
\usepackage[usenames]{color}
\usepackage{amsmath,dsfont,amssymb}
\usepackage{cite,url,fancybox,balance}
\usepackage{array,multirow}
\usepackage[active]{srcltx}
\usepackage{float}
\usepackage{bm}
\usepackage{cmll}
\usepackage{enumerate}

\newtheorem{lemma}{Lemma}

\newfont{\mbb}{msbm10 scaled 1100}

\definecolor{red}{RGB}{210,20,50}
\newcommand{\red}{\textcolor{red}}
\definecolor{lblue}{RGB}{30,160,240}
\newcommand{\lblue}{\textcolor{lblue}}
\definecolor{green}{RGB}{75,180,130}

\definecolor{blue}{RGB}{0,0,255}
\newcommand{\blue}{\textcolor{blue}}

\definecolor{magenta}{RGB}{255,0,255}
\newcommand{\magenta}{\textcolor{magenta}}
\definecolor{orange}{RGB}{255,128,0}
\newcommand{\orange}{\textcolor{orange}}
\definecolor{dgreen}{RGB}{0,151,0}

\def\ind{{\mathds{1}}}
\def\N{{\mathds{N}}} 
\def\R{{\mathds{R}}} 

\def\b0{{\bf 0}}

\renewcommand{\baselinestretch}{0.99}

\begin{document}

\title{Optimal relay location and power allocation for low SNR broadcast relay channels}

\author{
\authorblockN{Mohit Thakur}
\authorblockA{Institute for communications engineering,\\
Technische Universit\"{a}t M\"{u}nchen,\\
80290, M\"{u}nchen, Germany.\\
Email: mohit.thakur@tum.de}
\and
\authorblockN{Nadia Fawaz}
\authorblockA{Research Laboratory for Electronics,\\
Massachusetts Institute of Technology,\\
Cambridge, MA, USA.\\
Email: nfawaz@mit.edu}
\and
\authorblockN{Muriel M\'{e}dard}
\authorblockA{Research Laboratory for Electronics,\\
Massachusetts Institute of Technology,\\
Cambridge, MA, USA.\\
Email: medard@mit.edu}
}

\vspace{-6mm}
\maketitle
\vspace{-2mm}

\begin{abstract}
We consider the broadcast relay channel (BRC), where a single source
transmits to multiple destinations with the help of a relay, in the
limit of a large bandwidth. We address the problem of optimal relay
positioning and power allocations at source and relay, to maximize
the multicast rate from source to all destinations. To solve such a network planning problem, we develop a
three-faceted approach based on an underlying information theoretic
model, computational geometric aspects, and network optimization
tools. Firstly, assuming superposition coding and
frequency division between the source and the relay, the information
theoretic framework yields a hypergraph model of the wideband BRC, which captures the
dependency of achievable rate-tuples on the network topology. As the relay position varies, so does the set of hyperarcs constituting the hypergraph, rendering the combinatorial nature of optimization problem. We show that the convex hull C of all nodes in the 2-D plane can be divided into disjoint regions corresponding to distinct hyperarcs sets. These sets are obtained by superimposing all k-th order Voronoi tessellation of C. We propose
an easy and efficient algorithm to compute all hyperarc sets, and
prove they are polynomially bounded. Then, we circumvent the
combinatorial nature of the problem by introducing continuous switch
functions, that allows adapting the network hypergraph in
a continuous manner. Using this switched hypergraph approach, we
model the original problem as a continuous yet non-convex network
optimization program. Ultimately, availing on the techniques of
geometric programming and $p$-norm surrogate approximation, we
derive a good convex approximation. We provide a detailed
characterization of the problem for collinearly located
destinations, and then give a generalization for arbitrarily located
destinations. Finally, we show strong gains for the optimal relay
positioning compared to seemingly interesting positions.
\end{abstract}


\footnotetext{
ITMANET - 6915101: This material is based upon work under subcontract: 18870740-37362-C,
issued by Stanford University and supported by the Defense Advanced Research
Projects Agency (DARPA).}

\vspace{-4mm}
\section{INTRODUCTION} \label{sec:Introduction}
\vspace{-2mm}
Next-generation wireless standards, such as 3GPP Long Term
Evolution-Advanced (LTE-A) standard \cite{3GPP-TR2010}, propose
relays as a mean to extend cellular coverage or to increase data
rates. More specifically, LTE-A  defines relays of Type I as
coverage-extension relays which allow a base station (BTS) to reach
uncovered users in a cell, and relays of Type II as relays which allow to
increase the communication rate of a user already covered through a
direct link to the BTS \cite{3GPP-TR2010,Yang-Hu-ComMag2009}. In
terms of cellular deployment, a natural and practical question
arises as to where the relay node should be deployed.

In this paper, with the downlink of a cellular system with relays in
mind,  we address the aforementioned question for the broadcast
relay channel (BRC), which consists of a single source broadcasting
to multiple destinations with the help of a relay. In this paper, we
focus on  the wideband regime of wireless relay networks, also
denominated low signal-to-noise ratio (SNR) regime because power is
shared among a large number of degrees of freedom, making the
average SNR per degree of freedom low. We would like to point out
that addressing the low-SNR regime is relevant in next generation
cellular systems. Indeed, considering LTE, large bandwidths--- up to 20
MHz--- can be supported by all terminals (\cite{Astely-Dahlman-ComMag2009}). Due to power constraints in the low SNR regime, relays appear as a meaningful and natural way to increase rate and reliability.

Previous results on wireless systems in the low-SNR regime include
the capacity of point-to-point additive white Gaussian noise (AWGN)
channel \cite{Shannon-1949}, and multipath fading channel
\cite{Kennedy-1969,Gallager-1968,Telatar-Tse-2000,Medard-Gallager-2002,Verdu-2002,Subramanian-Hajek-2002,
Lun-Medard-AbouFaycal-2004}, both equal to the received SNR:
$C_{Fading}=C_{AWGN}= \frac{P}{N_0} = \displaystyle\lim_{W
\rightarrow \infty} W \log \big(1+\frac{P}{W N_0}\big)$; the
capacity of the multiple input multiple output (MIMO) channel
\cite{Zheng-Tse-2002,Ray-Medard-Zheng-2007}; the capacity region of
the AWGN broadcast channel (BC)
\cite{Cover-1972,ElGamal-Cover-1980,McEliece-Swanson-1987}, and AWGN
multiple access channel (MAC) \cite{Gallager-ITtrans1985}; and
bounds on the capacity of the non-coherent multipath fading relay
channel \cite{Fawaz-Medard-ISIT2010}. From these works, a conclusion
can be drawn on wireless systems in the low-SNR regime; the
major impairment in the low-SNR regime is neither multipath fading
nor interference, but noise, which is in contrast with the high-SNR
regime. Formulating the argument more concretely, in the presence of
multipath fading in the low-SNR, the same rates as the AWGN system
with the same received SNR can be achieved using non-coherent peaky
signals whereas spread-spectrum signals perform poorly. Moreover,
the low-SNR regime is not interference-limited: in particular, all
sources in the low-SNR MAC can achieve their interference-free
point-to-point capacity to the destination. Based on this
observation, the authors proposed in a recent work
\cite{Thakur-Medard-Globecom2010} an equivalent hypergraph model for
the low-SNR AWGN MAC and BC.  Then they used these models to build
an achievable hypergraph model for a more complex wireless network
with fixed sources, relays, and destinations, and showed that
optimizing power for maximizing multiple session rates boils down to
a straightforward linear program.

In this paper, we take a step forward by simultaneously optimizing the relay location and the power allocation, to maximize the multicast rate from a source $s$ to a set of destinations $T$.
Using concepts from information theory, computational geometry and network
optimization, we develop a comprehensive and efficient way to solve
this problem, which can be broadly divided into three parts:
\begin{enumerate}

\item \textbf{BRC hypergraph model: }
We propose a hypergraph model for the low-SNR BRC, which depends on
the topology of the network, essentially the placement of nodes on a
$2$-D plane. Given the source and destinations positions, computing
the hyperarcs in the BRC hypergraph model requires to get the
ordering of nodes in increasing distances from the source and relay,
for all relay positions (as the relay is not initially given). This problem
can be modeled as an \textit{ordered $k$-nearest neighbor problem},
for which we propose a solution based on superimposing the Voronoi
tessellations of all $k-1$ orders, where $k$ spans the destination
set.


\item  \textbf{Continuous hypergraph variations: }
For fixed source and destination positions, when the relay position
varies, the the network hypergraph changes accordingly rendering the problem combinatorial. Consequently, traditional
network optimization techniques cannot be applied directly, as they
assume a fixed given hypergraph. To circumvent this hard
combinatorial nature, we introduce continuous switch functions which
allow to change the network hypergraph in a continuous manner as the
relay position changes. Ultimately, this allows us to cast the
problem as a continuous optimization problem.

\item \textbf{Convex approximation:}
The resulting continuous network optimization problem is non-convex.
However, using geometric programming (GP) and $p$-norm approximation
techniques, we provide a good convex approximation of the original
problem to which standard convex optimization techniques can be
applied \cite{Boyd-book2004}. It should be noted that the problem is
NP-Hard in its original form mainly due to combinatorial nature and
continuous non-convex constraints.

\end{enumerate}

Hereafter, the paper is as follows. In Sections~\ref{sec:SysModel}
and III, we build the system model and formulate the general
problem, respectively. In  Section~\ref{sec:Collinear}, we solve the
problem for collinearly located destination nodes, and introduce
algorithms to compute distinct hypergraphs for various relay
positions. Section~\ref{sec:Arbitrary} extends to the general
problem case for an arbitrary topology, finally leading to the
conclusion in Section~\ref{sec:Conclusion}.

\color{black}
\section{Low {SNR} system model}\label{sec:SysModel}
Notations: $\N$ and $\R$ denote the sets of non-negative integers,
and real numbers, respectively. Let $m\in \N$ and $\N_m
\triangleq \{1,\ldots,m\}$ and Let $S$ be a set, the indicator function of $S$ is defined by $\ind_S(x)=1$ if $x\in S$, $\ind_S(x)=0$
if $x\notin S$.

In this section, we first recall the equivalent hypergraph models of
the wideband BC and MAC, then we use them to build an achievable
hypergraph model of the BRC, and finally we formulate the
optimization problem. In hypergraph models, a hyperarc $(u,v_1
v_2 \ldots v_k)$ of \emph{size $k$} connects a transmitter $u$ to an
ordered (increasing order of distance from transmitter $u$) set of $k$ receivers $\{v_1,v_2,\ldots,v_k\}$, all of which
can decode a message sent over the hyperarc equally reliably. Here, $u \notin \{v_1,v_2,\ldots,v_k\}$. Two
hyperarcs are disjoint if either they have different sources, or
different ordered receiver sets or both, e.g. $(u,v_{1}v_{2})$ is
disjoint from $(u,v_{2}v_{1})$ and $(u',v_{1}v_{2})$. Messages sent
over any pair of disjoint hyperarcs are independent. A hyperarc is
said to be \emph{activated} if its capacity is non-zero.

\subsection{Wideband BC and MAC model}\label{sec:BC-MAC}
\begin{figure*}[tp]
\begin{center}
\psfrag{s}[cc][cc]{{\small $s$}} \psfrag{d1}[cc][cc]{{\small $d_1$}}
\psfrag{d2}[cc][cc]{{\small $d_2$}} \psfrag{0}[tc][tc]{{\small $0$}}
\psfrag{Ps}[cc][cc]{{\small $P$}} \psfrag{h1}[cc][cc]{{\small
$h_1$}} \psfrag{h2}[cc][cc]{{\small $h_2$}}
\psfrag{N1}[cc][cc]{{\small $N_0$}} \psfrag{N2}[cc][cc]{{\small
$N_0$}} \subfigure[BC]{
\includegraphics[width=0.24\columnwidth]{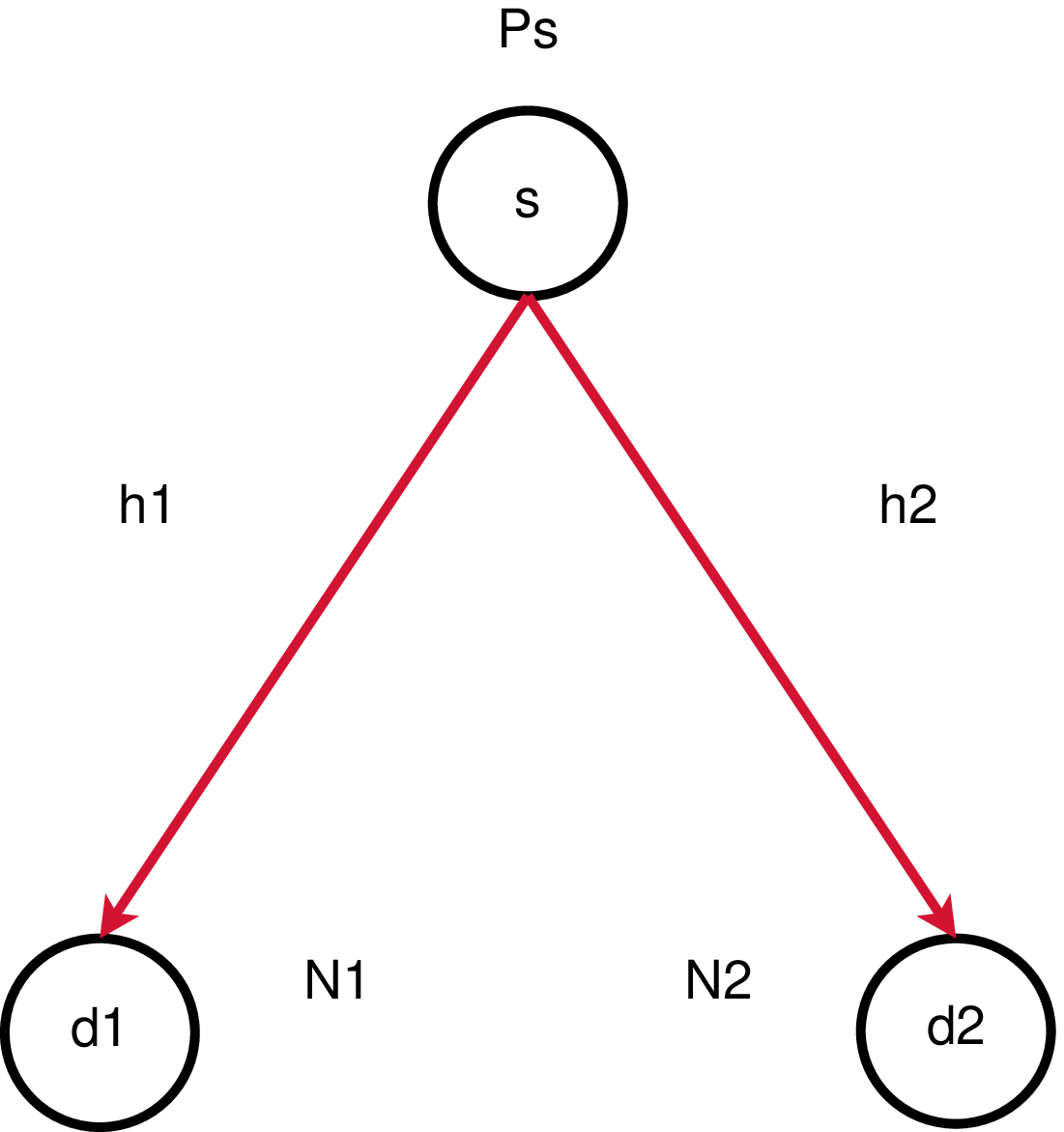}
\label{fig:BC} } \quad \psfrag{Cc}[cc][cc]{{\small \red{$R_c$}}}
\psfrag{Cb1}[cc][cc]{{\small \magenta{$R_1$}}}
\psfrag{Cb2}[cc][cc]{{\small \orange{$R_2$}}}
\subfigure[Wideband BC equivalent hypergraph]{
\includegraphics[width=0.26\columnwidth]{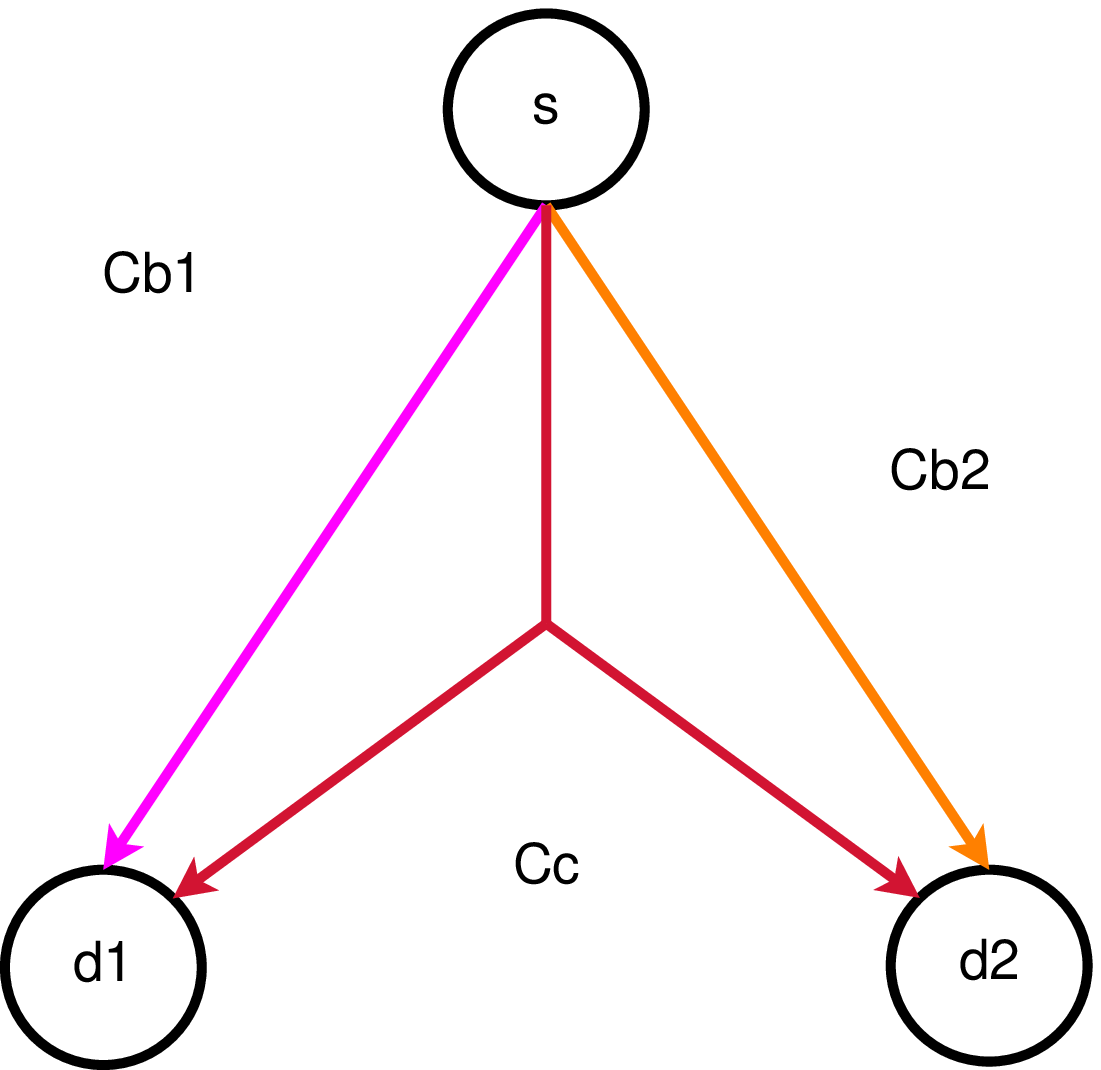}
\label{fig:BC-hypergraph} } \quad \psfrag{s1}[cc][cc]{{\small
$s_1$}} \psfrag{s2}[cc][cc]{{\small $s_2$}}
\psfrag{d}[cc][cc]{{\small $d$}} \psfrag{P1}[cc][cc]{{\small $P_1$}}
\psfrag{P2}[cc][cc]{{\small $P_2$}} \psfrag{N0}[cc][cc]{{\small
$N_0$}} \subfigure[MAC]{
\includegraphics[width=0.25\columnwidth]{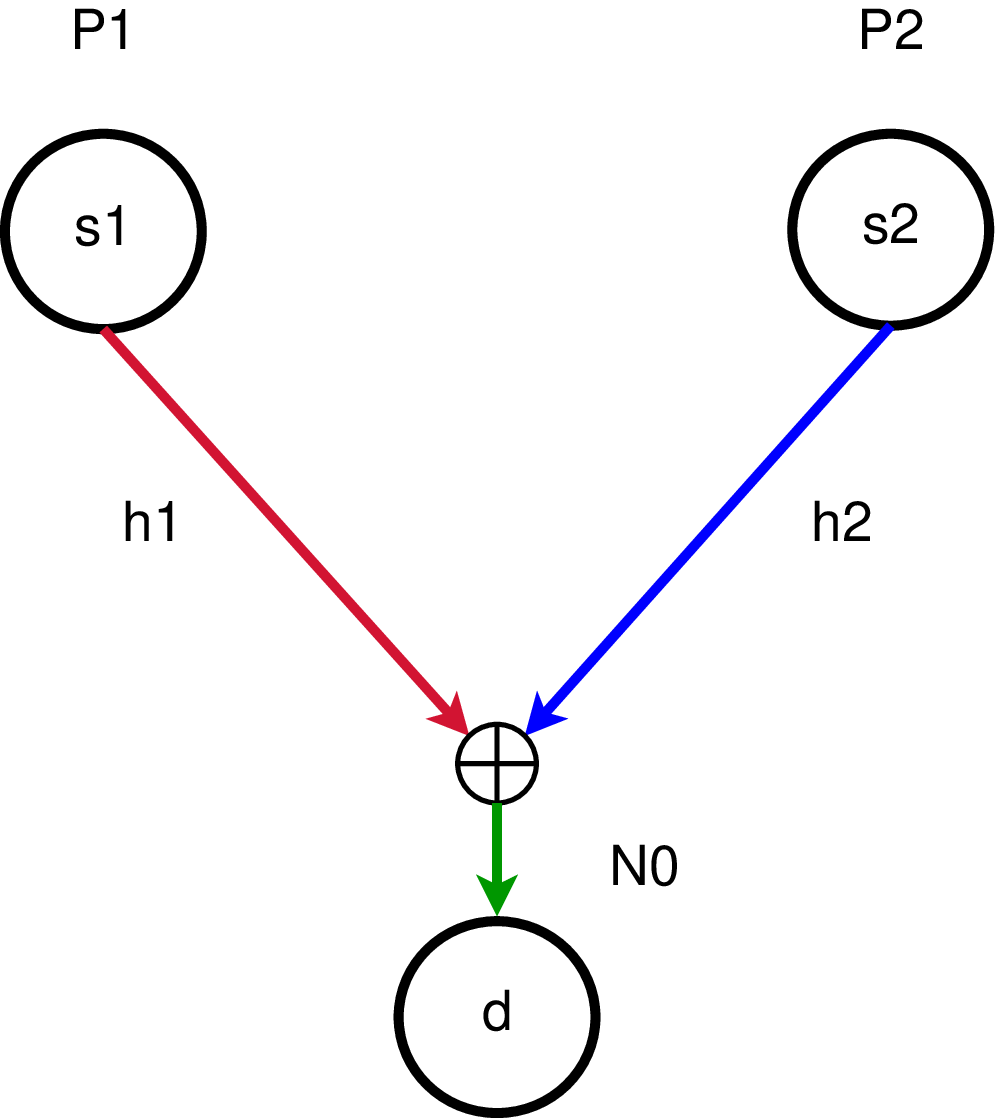}
\label{fig:MAC} } \quad \psfrag{R1}[tc][tc]{{\small $R_1$}}
\psfrag{R2}[br][br]{{\small $R_2$}} \psfrag{C1}[tc][tc]{{\small
\red{$h_1^2\frac{P_1}{N_0}$}}} \psfrag{C2}[cr][cr]{{\small
$C2=h_2^2\frac{P_2}{N_0}$}} \psfrag{C2bis}[tc][tc]{{\small
\blue{$h_2^2\frac{P_2}{N_0}$}}} \subfigure[Wideband MAC equivalent
hypergraph]{
\includegraphics[width=0.26\columnwidth]{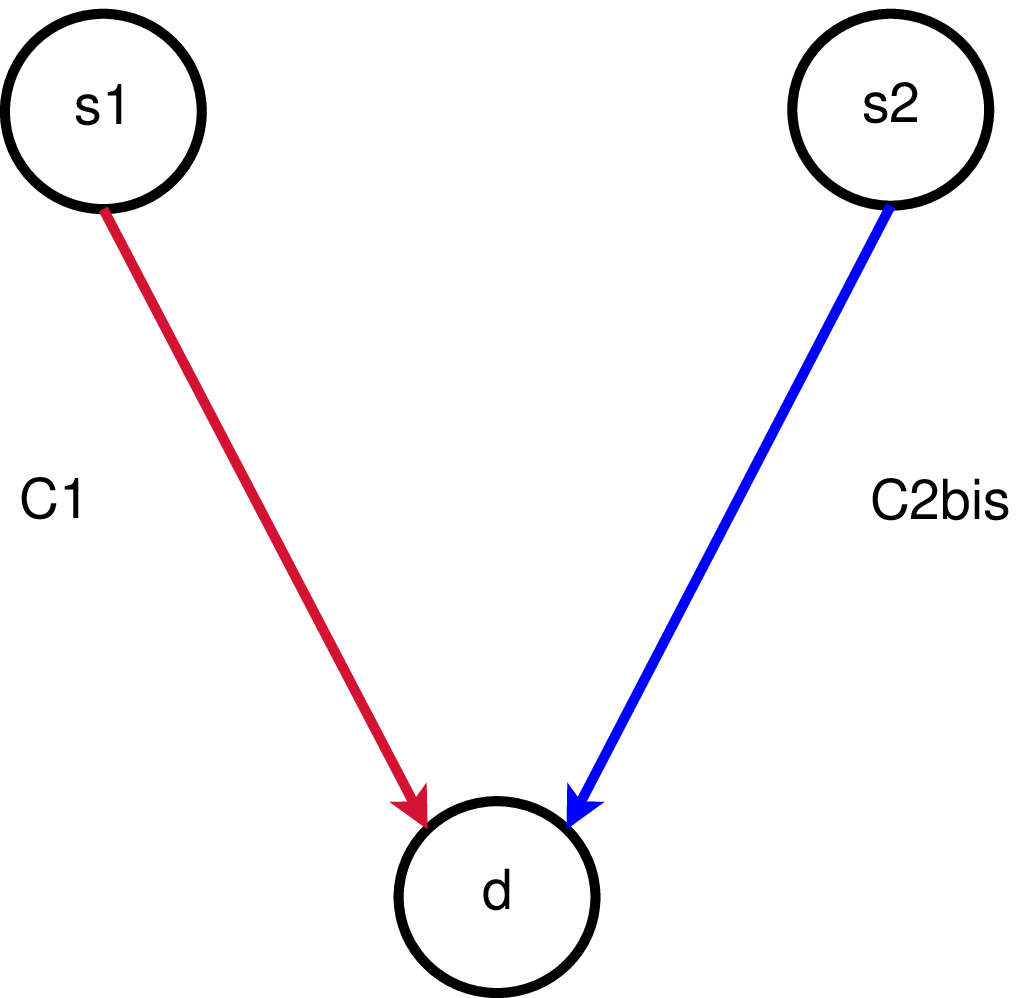}
\label{fig:MAC-hypergraph} } \quad \psfrag{s}[cc][cc]{{\small $s$}}
\psfrag{r}[cc][cc]{{\small $r$}} \psfrag{d1}[cc][cc]{{\small $d_1$}}
\psfrag{d2}[cc][cc]{{\small $d_2$}} \psfrag{Ps}[cc][cc]{{\small
$P_s$}} \psfrag{Pr}[cc][cc]{{\small $P_r$}}
\psfrag{N0}[cc][cc]{{\small $N_0$}} \psfrag{hs1}[cc][cc]{{\small
$h_{1s}$}} \psfrag{hs2}[cc][cc]{{\small $h_{2s}$}}
\psfrag{hsr}[cc][cc]{{\small $h_{rs}$}} \psfrag{hr1}[cc][cc]{{\small
$h_{1r}$}} \psfrag{hr2}[cc][cc]{{\small $h_{2r}$}} \subfigure[BRC]{
\includegraphics[width=0.29\columnwidth]{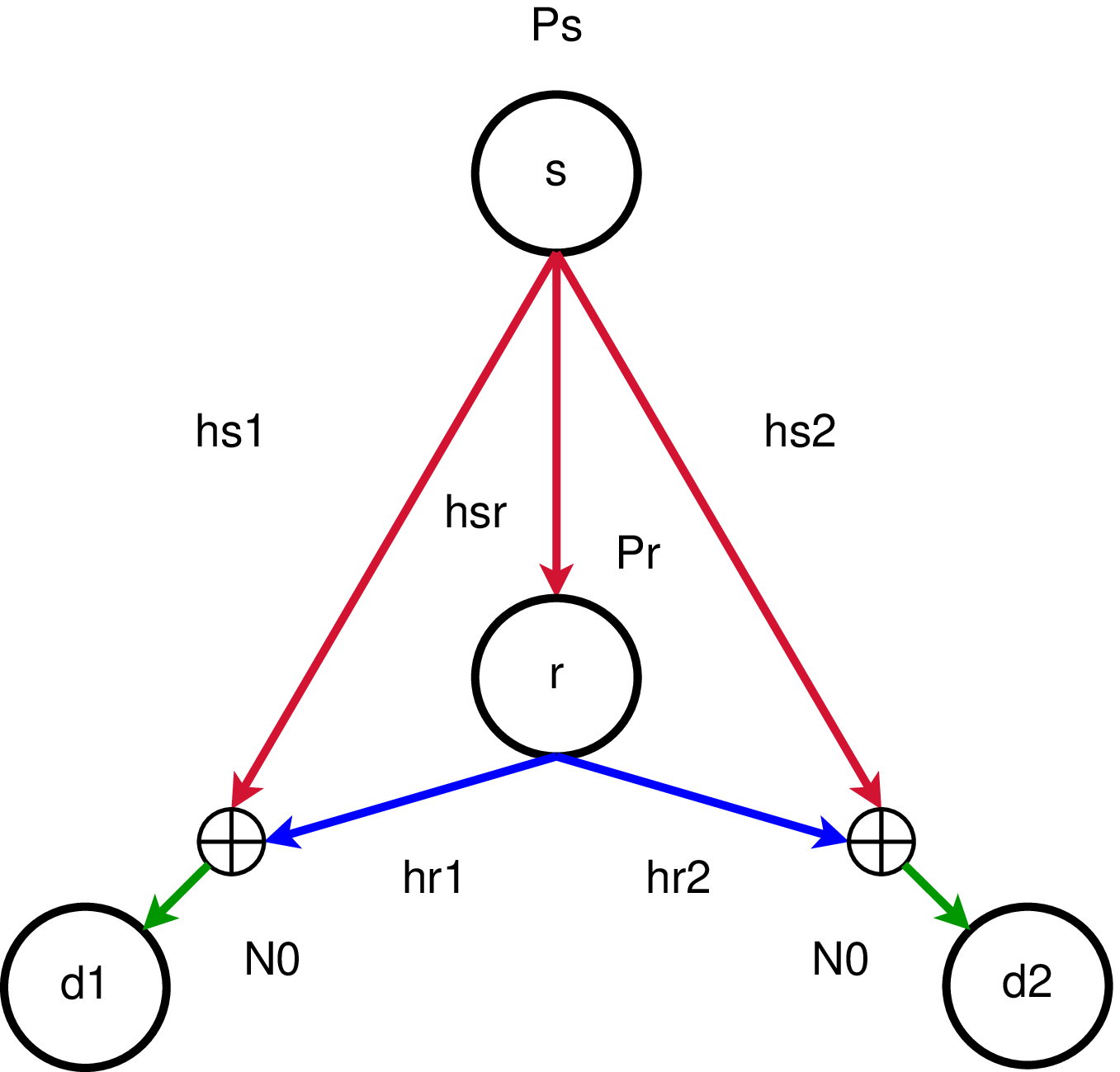}
\label{fig:BRC} } \quad \psfrag{R1}[cc][cc]{{\small
\magenta{$r_0$}}} \psfrag{R2}[cc][cc]{{\small \red{$r_1$}}}
\psfrag{R3}[cc][cc]{{\small \orange{$r_2$}}}
\psfrag{R4}[cc][cc]{{\small \blue{$r_3$}}}
\psfrag{R5}[cc][cc]{{\small \lblue{$r_4$}}}
\subfigure[Wideband BRC achievable hypergraph]{
\includegraphics[width=0.28\columnwidth]{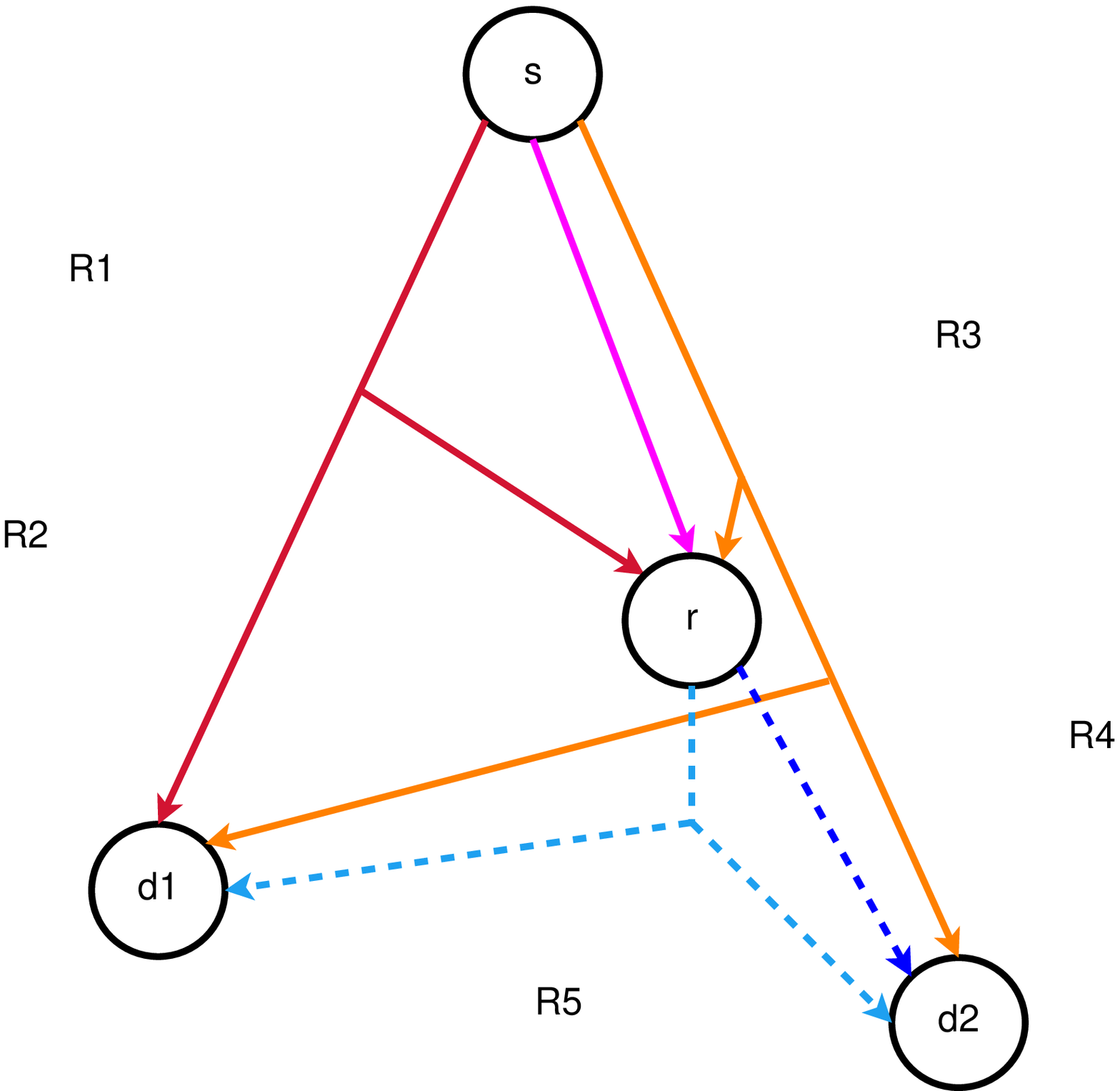}
\label{fig:BRC-hypergraph} } \label{fig:MUrateRegion}
\caption[Wideband MU Channels]{Wideband Multiple User Channels. The
BC rates are $R_1=(1-\alpha) h_1^2 \frac{P}{N_0}
\ind_{]h_2^2,+\infty[}(h_1^2)$, $R_2=(1-\alpha)  \frac{P}{N_0}
\ind_{[0,h_2^2[}(h_1^2)$, $R_c=\alpha \min\{h_1^2,h_2^2\}
\frac{P}{N_0}$. The BRC rates are $r_0=\frac{\alpha_0 P_s}{D_{sr}^2
N_0}$, $r_1= \frac{\alpha_1 P_s}{D_{s1}^2 N_0}$, $r_2=\frac{\alpha_2
P_s}{D_{s2}^2 N_0}$,$ r_3=\frac{\beta_1 P_r}{D_{r1}^2
N_0}$,$r_4=\frac{\beta_2 P_r}{D_{r2}^2 N_0}$}.
\end{center}
\vspace{-0.5cm}
\end{figure*}

Earlier results on multiple user channels show that BC and MAC
are not impaired by interference in the low SNR regime.

\subsubsection{Equivalent hypergraph of the wideband AWGN BC, \cite{Cover-1972,Cover-1998,ElGamal-Cover-1980,McEliece-Swanson-1987}}
Superposition coding is known to achieve the capacity region of the
AWGN BC. In the wideband limit, the rates achieved by superposition
coding boils down to the time-sharing rates, rendering time-sharing as optimal.

Consider the BC channel with source $s$, two destinations
$T=\{d_{1},d_{2}\}$ in Figure~\ref{fig:BC}, $D_{uv} > 0$ is the
distance of node $v$ from $u$, and let $\alpha \in [0,1]$ be the
power-sharing factor at source. Then both destinations can receive
the common rate $\alpha \min\{h_1^2,h_2^2\}\frac{P}{N_0}$, and the
most reliable destination can also receive a bonus private rate
$(1-\alpha) \max\{h_1^2,h_2^2\}\frac{P}{N_0}$, where $h$ is the path
loss factor, $N_{0}$ the channel noise and $P$ is the total source
power. This motivates the equivalent hypergraph model of the
wideband AWGN BC in Figure \ref{fig:BC-hypergraph}. The wideband BC
hypergraph model contains three hyperarcs: the common hyperarc from
the $s$ to $d_{1}$ and $d_{2}$ with capacity equal to the common
rate $R_{c}$, a private edge from $s$ to $d_1$ with capacity equal
to the private rate $R_{1}$ (if $d_1$ is more reliable than $d_2$,
i.e. $h_1^2 > h_2^2$, and to $0$ otherwise), and finally a private
edge $s$ to $d_2$ with a capacity equal to the bonus rate $R_{2}$ if $d_2$
is more reliable than $d_1$. Note that the two private hyperarcs
(associated with rates $R_{1}$ and $R_{2}$) cannot exist
simultaneously (as either $D_{s1}<D_{s2}$ or $D_{s2}<D_{s1}$):
thanks to the indicator functions in the capacity expressions, only
one of the private hyperarcs can be activated for a given topology.
In the general case of a wideband AWGN BC with $n$ destinations:
\begin{itemize}
\item For an arbitrary unknown topology, the full hypergraph model contains $2^n-1$ hyperarcs, from the source to every possible subset of destinations.
\item For a given known topology, only a subset of these hyperarcs are activated simultaneously.
Indeed, a given topology yields a given ordering of the destination
set in increasing order of reliability. Consequently, only $n$
hyperacs are simultaneously activated for a given topology:
one private arc of size $1$ to the most reliable destination i.e. $(s,d_{1})$, and
one common hyperarc of size $k$ to the $k$ most reliable
destinations for all $k\in\{2,...,n\}$ i.e. $(s,d_{1}..d_{k})$.
\end{itemize}

\subsubsection{Equivalent hypergraph of the wideband AWGN MAC \cite{Gallager-ITtrans1985}}
Consider two sources $s_{1}$ and $s_{2}$ and a single destination
$d$ in Figure~\ref{fig:MAC}. In the wideband regime, the large
number of degrees of freedom renders negligible interference, and
allows all sources to achieve their point-to-point capacity to the
destination, as with frequency division multiple access (FDMA). Thus, the
respective capacities of $s_{1}$ and $s_{2}$ are $C_1= h_1^2\frac{P_1}{N_0}$ and $C_2=
h_2^2\frac{P_2}{N_0}$. As shown in Figure~\ref{fig:MAC-hypergraph},
the equivalent hypergraph model contains only two edges, one from
$s_1$ to $d$ with capacity $C_1$ and one from $s_2$ to $d$ with
capacity $C_2$. In the general wideband MAC with $n$ sources, the
hypergraph model consists of $n$ hyperarcs of size $1$ with non-zero
capacity, from each source to the destination.

\subsection{Wideband BRC model}

Consider the broadcast relay channel in Figure $1$ (e), where a source s transmits to a set of $n$ destinations $T=\{d_i\}_{i\in \N}$ with the help of a relay $r$. We assume that all nodes are equipped with a single antenna. The
source and the relay have given respective average power constraints
$P_s$ and $P_r$, and they transmit in two different frequency bands,
$W_s$ and $W_r$ respectively, so as to respect the half-duplex
constraints at the relay. During each time slot, $s$ transmits a new
codeword which is received by the relay $r$ relay and $T$ (all
destinations); $r$ processes the signal received from $s$ in
the previous time slot and retransmits it to $T$; the destinations
use the signals they received directly from $s$ and through $r$ to
decode a new codeword.

The wireless link between two nodes $u \in\{s,r\}$ and
$v \in\{r,d_1,\ldots,d_n\}$ is modeled by an AWGN channel. In other words, when node $u$ transmits a
signal $x_u(t)$, node $v$ receives a signal $y_{vu}(t)=h_{vu}
x_u(t)+z_v(t)$ where $h_{vu}=\frac{1}{D_{vu}^ {\alpha/2}}$ is an
attenuation coefficient modeling pathloss, and  $z_v(t)$ is a white
Gaussian noise process with power spectral density $N_0$. Note that
although we consider AWGN channels, the low-SNR analysis could be
extended to multipath fading channels: indeed, it was shown in
\cite{Fawaz-Medard-ISIT2010} that in the wideband multipath fading
relay channel, the same rates can be achieved as in the wideband
AWGN relay channel with the same average received SNR on each link.

The AWGN BRC consists of two BC components in series for $s$ and $r$, the BC from
$s$ to $\{r,d_1,\ldots,d_n\}$ in red in Figure~\ref{fig:BRC}, and
the BC from $r$ to $\{d_1,\ldots,d_n\}$ in blue--- and of $n$
parallel MAC components, such as the MAC from $\{s,r\}$ to $d_1$
represented by the sum of the red and blue lines resulting into a green line.

As in \cite{Thakur-Medard-Globecom2010}, we make the assumption that
the source $s$ and the relay $r$ are constrained to transmit using
the scheme that would be optimal for their respective wideband
BC-components: $s$ transmits using superposition coding in the
source band $W_s$; $r$ decodes the messages it received from $s$,
and then retransmits using superposition coding in the relay band
$W_r$; each destination $d_i$ decodes by using the interference-free
signals it received from $s$ and $r$.

Under these constraints on the communication scheme, the resulting
hypergraph model of the BRC \cite{Thakur-Medard-Globecom2010} is
simply the concatenation of the equivalent hypergraphs of the
BC-components and the MAC components. We will denote this hypergraph
as $\mathcal{G(N,H)}$, where, $\mathcal{N} \supset
T=\{d_{1},..,d_{n}\}$ and $\mathcal{H}$ is the set of hyperarcs. The
set $\mathcal{H}$ can be partitioned into two disjoint sets: one is
$\mathcal{H}_s$ of source hyperarcs emanating from $s$, and the
set  $\mathcal{H}_r$ of relay hyperarcs emanating from the $r$,
where $\mathcal{H}_{s} \cup \mathcal{H}_{r}=\mathcal{H}$.
Figure~\ref{fig:BRC-hypergraph} illustrates the hypergraph in the
case of a given topology with two destinations. In this figure, we
assume that $r$ is the closest node to $s$, followed by $d_1$ and
then $d_2$, and we show only the activated hyperarcs.


It should be pointed out that this BRC hypergraph model is only an
achievable model, and not an equivalent model. Indeed, in the case
of a single destination, the BRC boils down to the relay channel,
and it was shown in \cite{Fawaz-Medard-ISIT2010} that with a
different coding scheme, it is possible to achieve a higher rate in
the wideband relay channel than any rate obtained by the scheme in
\cite{Thakur-Medard-Globecom2010}. Thus the BRC hypergraph model
proposed in this paper provides only an achievable rate region, but
not the full rate region of the BRC. However the relaying scheme in
\cite{Thakur-Medard-Globecom2010}, and the associated hypergraph
model, have the benefit to easily extend to large complex network.

\section{General problem structure}
Given a topology of the set of nodes $\mathcal{N} \backslash r$, and the aforementioned achievable hypergraph model of the wideband BRC, we recall the following question: \emph{What is the optimal relay position and power allocations at $s$ and $r$, that maximize the multicast rate $R_m$ from $s$ to the destination set $T$?} Here, the multicast rate is the rate experienced by the least reliable destination in the set $T$, and is given by its min-cut.
%
%
%
To solve the problem in full generality, we propose a two-stage approach, as follows:
\begin{itemize}

\item \underline{Pre-processing}: The pre-processing stage computes all distinct $\mathcal{H}_{s}$ and $\mathcal{H}_{r}$, respectively, for all positions of the relay inside the region being considered on the $2$-D plane, given by hyperarc sets $\mathcal{H}_{s}$ and $\mathcal{H}_{r}$, respectively. Since, only a subset of  these hyperarcs are active when the relay is in a certain region, we associate each hyperarc $(u,V)\in\mathcal{H}$, along with a continuous switch function $f_{uV}$. The switch function $f_{uV}$ activates the hyperarc $(u,V)$ by taking the value $1$ when it should exist and deactivates the hyperarc by taking value $0$ when it should not exist. In this section and section IV, we devise efficient algorithms to compute all the distinct hyperarcs. Once the hyperarc set $\mathcal{H}$ is constructed, we can then compute all the possible paths from the source to each destination. The total number of paths from $s$ to a destination $d_i \in T$ will be denoted $K_i$, and the rates on these paths will be written $\{r_1^{d_i}, ... , r_{K_i}^{d_i} \}$.
\vspace{1mm}

\item \underline{Optimization}: The second stage involves solving a network flow optimization problem. After obtaining $\mathcal{H}$, the multicast rate maximization problem for $\mathcal{G(N,H)}$ can be formulated as:
\vspace{-2mm}
    \begin{eqnarray}
    & \mbox{Program (A): } & \mbox{maximize }  \left( R_{m} \right)   \nonumber\\
    & \mbox{subject to: } &  R_{m} \leq r_{i} \mbox{, } \: \forall i \in \N_n,  \label{eq:ProgA-1}\\
    &  & r_{i} \leq  \sum_{k=1}^{K_i} r_{k}^{d_i}  \mbox{, }  \: \forall i \in \N_n, \label{eq:ProgA-2}\\
    &  &\hspace{-30mm} \max_{\substack{((i,k_{i})| k_{i} \in (u,V), \\ k_{i} \in [1,K_{i}])}} r^{d_{i}}_{k_{i}} \leq y_{uV} \mbox{, } \:  i\in \N_{n}, \forall (u,V) \in \mathcal{H}, \hspace{4mm} \label{eq:ProgA-3}\\
    &  & \hspace{-10mm} y_{uV} \leq c_{uV} f_{uV} \mbox{, } \: \forall (u,V) \in \mathcal{H}, \label{eq:ProgA-4}\\
    & \hspace{-10mm} \mbox{where,}& \hspace{-10mm} c_{uV} \in C_{uV}, \forall (u,V) \in \mathcal{H}. \label{eq:ProgA-5}
    \end{eqnarray}
    Here (\ref{eq:ProgA-1}) implies $R_{m}$ is the minimum among the total rates experienced at all destinations, (\ref{eq:ProgA-2}) says that the rate for destination $d_{i}$ is the sum of rates on all the paths from $s$ to $d_{i}$. (\ref{eq:ProgA-3}) captures the network coding constraint, and the switch function $f_{uV}$ in (\ref{eq:ProgA-4}) activates and deactivates hyperarc $(u,V)$ as the optimization algorithm goes from one relay position to the other to maximize  $R_{m}$. (\ref{eq:ProgA-5}) implies that the capacity of hyperarc $(u,V) \in \mathcal{H}$ is determined by implicit constraints of power and distance.

\end{itemize}
Note that, it is because of continuous switch functions that we have
a continuous optimization problem, else we would need combinatorial
constraints to capture the right hyperarcs that need to be activated for each relay position.

In the sequel of this section, we describe the two stages in detail. At first, we show that the optimal relay
position lies in the convex hull $\mathcal{C}$ of nodes $\mathcal{N} \backslash r = \{s, T\}$.

\subsection{Convex hull lemma}
\vspace{-2mm}
We prove hereunder an intuitive lemma which helps in formulating
the problem in a geometric sense.

\begin{figure}
   \begin{center}
   \psfrag{s}{$s$}
   \psfrag{r}{$r$}
   \psfrag{1}{$d_{1}$}
   \psfrag{2}{$d_{2}$}
   \psfrag{s1}{}
   \psfrag{s1'}{}
   \psfrag{s2}{}
   \psfrag{s2'}{}
   \psfrag{12}{}
   \psfrag{12'}{}
    \includegraphics[width=0.4\columnwidth]{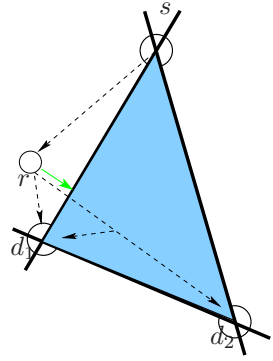}
    \vspace{-2mm}
\end{center}
\caption{{Three line line segments joining $(s,d_{1},d_{2})$ form the convex hull $\mathcal{C}$ (blue region). The hyperarcs $(s,r)$, $(r,d_{1})$ and $(r,d_{1}d_{2})$ are shown with dashed arrows. }} \label{fig:ConvexHull} \vspace{-6mm}
\end{figure}

\begin{lemma}\label{lem:ConvexHull}
Given $\mathcal{G(N,H)}$, the relay location that maximizes the
multicast rate $R_{m}$ from $s$ to $T$ lies inside the convex hull
of the nodes $\mathcal{N} \backslash r$.
\end{lemma}

\begin{proof}
We prove this lemma by building an intuitive argument. To start
with, consider the four node system of $\{s,r,d_{1},d_{2}\}$ in
Figure~\ref{fig:ConvexHull}. The only node that is allowed to take
its desired location is the relay $r$. The convex hull $\mathcal{C}$
of the node set $\mathcal{N} \backslash r =\{s,d_{1},d_{2}\}$ is
shown as the shaded area in Figure~\ref{fig:ConvexHull}. Consider the
arbitrary position outside $\mathcal{C}$ in the above scene e.g. $r$
is placed outside $\mathcal{C}$. The hyperarc sets for this position
of $r$ are given by $H_{s}=\{(s,r),(s,r,d_{1}),(s,r,d_{1},d_{2})\}$
and $H_{r}=\{(r,d_{1}),(r,d_{1},d_{2})\}$, where
$\mathcal{H}$, assuming that the ordered sets of nodes in increasing distances from $s$ and $r$ are given by $(s,r,d_{1},d_{2})$ and $(r,d_{1},d_{2})$, respectively.

Consider the rates on the path $\{(s,r),(r,d_{1})\}$ from $s$ to
$d_{1}$, which are given by:
\begin{equation}\label{eq:ConvexHullRates}
 R_{sr} \leq \frac{P_{sr}}{D_{sr}^{\alpha}N_{0}}, \hspace{2mm} R_{r1} \leq \frac{P_{r1}}{D_{r1}^{\alpha}N_{0}},
\end{equation}
where, $P_{s}$ and $P_{r}$ are source and relay powers. It is clear
that by moving $r$ towards the boundary of $\mathcal{C}$, i.e. line
segment joining $s$ and $d_{1}$, total rate on this path
($\min(R_{sr},R_{r1})$) could be increased. The triangle's
inequality corroborates this fact,

\begin{equation}\label{eq:trineq1}
D_{sr} + D_{rd_{1}} \geq D_{sd_{1}},
\end{equation}
\begin{equation}\label{eq:trineq2}
\begin{split}
(R_{sr} + R_{r1})&|_{\{D_{sr} + D_{rd_{1}} > D_{sd_{1}}\}} \\
& < (R_{sr} + R_{r1})|_{\{D_{sr} + D_{rd_{1}} = D_{sd_{1}}\}}.
\end{split}
\end{equation}
 (\ref{eq:trineq1}) is the triangles inequality for the triangle $\triangle_{srd_{1}}$ and this implies (\ref{eq:trineq2}), which states that the rate from $s$ to $d_{1}$ on the path $\{(s,r),(r,d_{1})\}$ could be increased by simply bringing the $r$ to towards the line segment joining $s$  and $d_{1}$.
 It is straightforward to see that this also increases the rate for all other receivers in system (consider triangles $\triangle_{srd_{i}}$).

This reasoning can easily be generalized to any arbitrary position of the relay outside the convex hull $\mathcal{C}$, and to any arbitrary number of destinations $|T| > 2$. Thus, we conclude that for any
given instance of BRC, the relay location maximizing the multicast rate lies inside or on the border of the convex
hull of $\{s + T\}$ nodes, but never outside. Hence, proved.
\end{proof}
\vspace{1mm}
Lemma~\ref{lem:ConvexHull} implies that we only need to
consider relay locations in the convex hull $\mathcal{C}$ of nodes $\mathcal{N} \backslash r$,
to maximize the multicast rate. There are efficient algorithms for
constructing a convex hull of $n+1$ points, c.f.
\cite{Preparata-book1985} and references within.

\subsection{Pre-processing algorithms}
\vspace{-1mm} The pre-processing stage consists of three sub-stages.
For all relay positions in $\mathcal{C}$, the first sub-stage
computes the source hyperarc set $\mathcal{H}_s$, the second
sub-stage computes the relay hyperarc set $\mathcal{H}_{r}$, and the
last sub-stage computes all the source-destination paths $\forall
d_{i} \in T$. Now, we develop algorithms to compute
$\mathcal{H}_{s}$ and $\mathcal{H}_{r}$ and give upper-bounds on the
number of distinct hyperarcs in $\mathcal{H}_{s}$ and
$\mathcal{H}_{r}$.

\begin{figure}[tp]
   \begin{center}
\psfrag{s}{$s$}
\psfrag{1}{$d_{1}$}
\psfrag{2}{$d_{2}$}
\psfrag{3}{$d_{3}$}
\psfrag{a}{(a)}
\psfrag{b}{(b)}
   \psfrag{hsr123}{$sr123$}
   \psfrag{hs1r23}{$s1r23$}
   \psfrag{hs12r3}{$s12r3$}
\psfrag{sr}{$(s,r),$}
\psfrag{sr1}{$(s,rd_{1}),$}
\psfrag{sr12}{$(s,rd_{1}d_{2}),$}
\psfrag{sr123}{$(s,rd_{1}d_{2}d_{3})$}
\psfrag{s1}{$(s,d_{1})$}
\psfrag{s1r}{$(s,d_{1}r)$}
\psfrag{s12}{$(s,d_{1}d_{2})$}
\psfrag{s12r}{$(s,d_{1}d_{2}r)$}
    \includegraphics[width=0.44\textwidth]{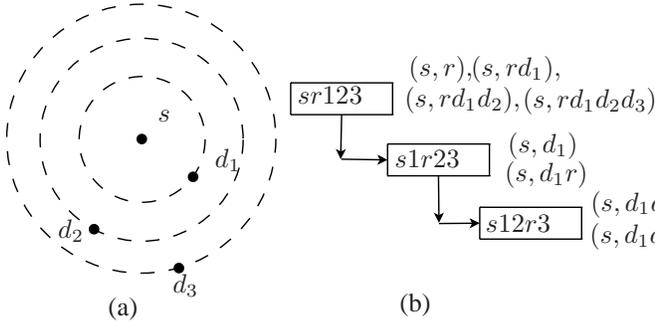}
\end{center}
\vspace{-4mm} \caption{(a): $T=\{d_{1},d_{2},d_{3}\}$ ordered set
$T$ with three concentric circles $c_{1}$ for each destination. (b): Shows the ordered set of nodes w.r.t. $r$ for the disc and two rings. Starting from the disc, the $\mathcal{H}_{s}$ is computed and then for each ring the two new hyperarcs are added.}
\label{fig:SourceHyperarcAlgo} \vspace{-6mm}
\end{figure}

\subsubsection{Source hyperarcs ($\mathcal{H}_{s}$)}\label{sec:SourceHyperarcs}

We first prove a lemma on $|\mathcal{H}_s|$.

\begin{lemma}\label{lem:SourceHyperarcs}
The number of distinct source hyperarcs inside the convex hull
$\mathcal{C}$ is upper bounded by $3n-1$, where $n=|T|$.
\end{lemma}

\begin{proof}
Consider a BRC with $s$ and an ordered set of three destinations
$T=\{d_{1},d_{2},d_{3}\}$ ($|T|=n=3$). Let $c_i$ be the circle
centered at $s$ passing though $d_i$. An example is
illustrated in Figure~\ref{fig:SourceHyperarcAlgo}(a), the three
circles partition the $2$-D plane into rings and discs that are given by disc $C_{1}$ and two
concentric rings $\mathcal{R}_{21}$ and $\mathcal{R}_{32}$.

For the positions of $r$ inside these areas there are distinct sets
of source hyperarcs. Computing the source hyperarcs when $r$ is inside these regions and
starting with the disc $c_{1}$, we simply get a set of $4$ hyperarcs
$H_{s}=\{(s,r),(s,rd_{1}),(s,rd_{1}d_{2}),(s,rd_{1}d_{2}d_{3})\}$.
Each time $r$ crosses the border of circle $c_{i}$ and enters the ring
$\mathcal{R}_{ii-1}$, there are two hyperarcs that change and hence the
new hyperarcs must be added, i.e. $\{(s,d_{1}),(s,d_{1}r)\}$ for
$\mathcal{R}_{21}$ and $\{(s,d_{1}d_{2}),(s,d_{1}d_{2}r)\}$ for
$\mathcal{R}_{32}$, (refer Figure~\ref{fig:SourceHyperarcAlgo}(b)). This is due to the fact that, when $r$ enters a new region, the ordered set of nodes in increasing distances from $s$ of the new region is different from the previous region in only two places, e.g. for the disc $C_{1}$ and ring $\mathcal{R}_{21}$ the ordered sets are given by $(s,r,d_{1},d_{2},d_{3})$ and $(s,d_{1},r,d_{2},d_{3})$, respectively.

Thus, the maximum number of distinct source hyperarcs that can exist for all relay positions in
$\mathcal{C}$ is given by $(n+1)+2(n-1)=3n-1$. Hence, proved.
\end{proof}

\vspace{1mm}
At this point we would like to highlight a couple subtleties:
\begin{enumerate}[(a)]
    \item If some $d_{i} \in T$ are equidistant from $s$, then their respective circles coincide, hence reducing the number of disjoint rings, and the number of distinct hyperarcs becomes less than $3n-1$. Thus, Lemma~\ref{lem:SourceHyperarcs} is an upper bound.
    \item The simple algorithm outlined in \emph{Lemma~\ref{lem:SourceHyperarcs}} builds all the possible hyperarcs efficiently in the sense that only distinct hyperarcs are added to $\mathcal{H}_{s}$  along with their respective switch functions, thus avoiding any redundancy.
\end{enumerate}

\underline{\emph{Switch Function:}} The activation/deactivation of a
hyperarc can be performed by the switch function. For instance, the
switch function associated with hyperarc $(s,rd_{1})$ (Figure~\ref{fig:SourceHyperarcAlgo}) is
$f_{sr1}=(1+\gamma e^{(- \gamma z_{sr1})})^{-1}$, where
$z_{sr1}=D_{s1}-D_{sr}$ and $\gamma >> 1$. When $r$ is in $C_1$,
$z_{sr1}$ is positive, and $f_{sr1}\simeq 1$, thus hyperarc
$(s,rd_{1})$ is active. Similarly, when $r$ is outside $C_1$,
$z_{sr1}$ is negative, and $f_{sr1}\simeq 0$, thus hyperarc
$(s,rd_{1})$ is deactivated. Notice, that for the hyperarcs inside a
ring (e.g. $\mathcal{R}_{ii-1}$), the switch function will be a product
of two functions, each for the regions of concentric circles that make this
ring.

\begin{figure}
\begin{center}
\psfrag{s}[cc][cc]{{\tiny $s$}} \psfrag{d1}[cc][cc]{{\tiny $d_1$}}
\psfrag{d2}[cc][cc]{{\tiny $d_2$}} \psfrag{d3}[cc][cc]{{\tiny
$d_3$}} \psfrag{d4}[cc][cc]{{\tiny $d_4$}}
\psfrag{(142)}[cc][cc]{{\tiny $(1423)$}}
\psfrag{(124)}[cc][cc]{{\tiny $(1243)$}}
\psfrag{(214)}[cc][cc]{{\tiny $(2143)$}}
\psfrag{(213)}[cc][cc]{{\tiny $(2134)$}}
\psfrag{(231)}[cc][cc]{{\tiny $(2314)$}}
\psfrag{(234)}[cc][cc]{{\tiny $(2341)$}}
\psfrag{(324)}[cc][cc]{{\tiny $(3241)$}}
\psfrag{(342)}[cc][cc]{{\tiny $(3421)$}}
\psfrag{(432)}[cc][cc]{{\tiny $(4321)$}}
\psfrag{(431)}[cc][cc]{{\tiny $(4312)$}}
\psfrag{(413)}[cc][cc]{{\tiny $(4132)$}}
\psfrag{(412)}[cc][cc]{{\tiny $(4123)$}}
\psfrag{(421)}[cc][cc]{{\tiny $(4213)$}}
\psfrag{(241)}[cc][cc]{{\tiny $(2413)$}}
\psfrag{(243)}[cc][cc]{{\tiny $(2431)$}}
\psfrag{(423)}[cc][cc]{{\tiny $(4231)$}}
\includegraphics[width=0.98\columnwidth]{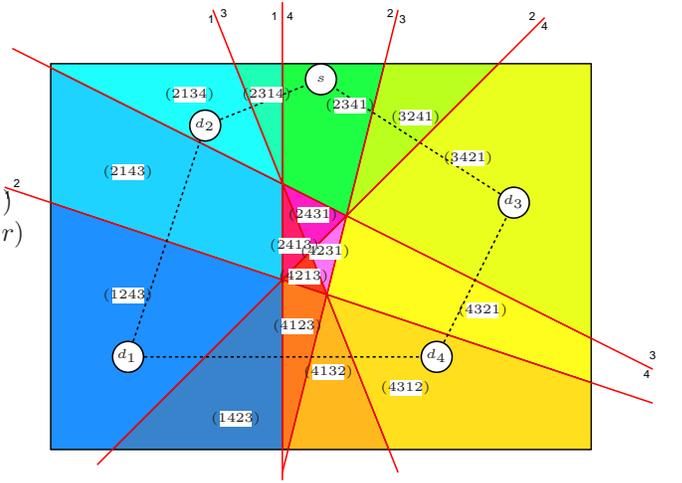}
\end{center}
\vspace{-4mm}
\caption{{The 4 destinations are shown with the cells showing
ordered pairs (in increasing distance) in the plane after
superimposing $1$st,  $2$nd and $3$rd order Voronoi diagrams. The dashed polygon forms $\mathcal{C}$.}}
\label{fig:Voronoi}
\vspace{-6mm}
\end{figure}

\subsubsection{Relay hyperarcs ($\mathcal{H}_{r}$)}\label{sec:RelayHyperarcs}

The relay hyperarcs are determined by
the ordering of node set $T$ with increasing distance from $r$. Thus, we need to partition the 2-D plane into
disjoint regions where the ordering of $T$ with respect to $r$ stays
the same. This is equivalent to computing the order-$k$ Voronoi
tessellations of $\mathcal{C}$ for all $k \in \N_{n-1}$,
and then superimposing them to get the ordered set $T$ of nearest
neighbors in $\mathcal{C}$, \cite{Edelsbrunner-book1987}.  It is
known that the superimposition of Voronoi tessellations results in
convex disjoint areas (polygons in our case). Once the ordered set $T$ for each disjoint region is computed, the
$n$ activated relay hyperarcs for each region are obtained the same way as in
Section~\ref{sec:BC-MAC} for the case of a BC with a given known
topology. Thus, the algorithm outlined in Lemma 1 could be used to generate $\mathcal{H}_{r}$. Figure~\ref{fig:Voronoi}, illustrates the superimposed
disjoint regions of ordered destinations with respect to $r$ for $n=4$. The simplest
way to compute these regions (ref. \cite{Lee-1982}) is to draw the
perpendicular bisector of every destination pair $(d_i,d_j)$ in $T^{2}$.
This method has the complexity of order $O(n(n-k))$, where $k \in
\N_{n-1}$. Hence, we obtain the partitions of $\mathcal{C}$ for
distinct ordering of the set $T$ with respect to $r$ from which we can generate $\mathcal{H}_{r}$.

The pre-processing in almost all network planning problems is
computationally heavy as there is plenty of time up-front compared to
real-time applications. In our case it fits better as all the
computations are of polynomial order.

In the next sections, we formulate the problem of optimal relay
positioning as a non-convex network flow optimization problem and
provide a good convex approximation. For simplicity and clarity in understanding, we divide the
problem into two cases. The first case is for collinearly located
destination nodes and the second case is for arbitrarily
located destinations.

\section{COLLINEAR CASE} \label{sec:Collinear}
In this section we develop the method to solve the
simpler version of the problem where the destination nodes are
collinearly located. It helps understanding the main concepts and underlying algorithms, and ultimately leads to the solution for the arbitrary case.

\subsection{Pre-processing}

\subsubsection{Convex hull $\mathcal{C}$}

For the collinearly located destination set $T$, the set $T$ could
be ordered in the increasing order of abscissa, for instance with the left
most node being at the origin. Hereafter, we assume that the
left-most (respectively rightmost) node $d_l \in T$ ($d_r \in T$) is situated
at the origin (respectively right most at horizontal axis), and the
rest of the destination nodes are on the positive $x$-axis, and
finally that the source is in the positive quadrant. Note that, $s$
could be the leftmost node compared to any $d_{i} \in T$, in this
case $s$ could be assumed to be on the positive vertical axis and the
set $T$ would accordingly be placed on positive horizontal axis with $d_{l}$ not being at origin.
Since all $d_{i} \in T$ are collinearly located, the convex hull
will always be the triangle $\triangle_{sd_{l}d_{r}}$ (ref.
Figure~\ref{fig:edge-a}). Thus, $\mathcal{C}$ is always given by only three
inequalities in this case.

\subsubsection{Source hyperarcs}

As explained in Section~\ref{sec:SourceHyperarcs}, the source
hyperarcs are functions of the source-destination distances.
Consequently, the algorithm outlined in
Lemma~\ref{lem:SourceHyperarcs} could be used to compute
$\mathcal{H}_{s}$.

\subsubsection{Relay hyperarcs}
Also, as shown in Section~\ref{sec:RelayHyperarcs}, we compute the
perpendicular bisectors of every destination pair $(d_i,d_j) \in T^{2}$,
to compute the superimposed convex disjoint regions of
ordered destination sets w.r.t. $r$ (ref. Figure~\ref{fig:edge-b}).
For the collinearly located nodes, the computation of the set
$\mathcal{H}_{r}$ is greatly simplified due to parallelism of all bisectors. The
following Lemma is just an easy and straightforward formalization.

\begin{lemma}\label{lem:RelayHyperarcs}
For $|T|=n$ collinear destinations, the total number of distinct
relay hyperarcs in $\mathcal{C}$ is upper bounded by
$n^{2}=n+2(\beta-1)$, where $\beta=\binom{n}{2} +1$ is the number of
bisected regions.
\end{lemma}

\begin{proof}
Given a set $T$ of n colinear destinations, the maximum number of bisectors
are given by $\binom{n}{2}$. Then the total number of bisected
regions is given by $\beta=\binom{n}{2}+1$, as shown in
Figure~\ref{fig:edge-b}. Since, crossing each bisector only changes
two nodes in the ordered destination set, using the algorithm
outlined in \emph{Lemma~\ref{lem:SourceHyperarcs}} (ref.
Figure~\ref{fig:SourceHyperarcAlgo}), we can compute all distinct
relay hyperarcs to $T$, $|\mathcal{H}_{r}| = n+2(\beta-1)=n^{2}$.
Hence, proved.
\end{proof}

With only little more formalization of Lemma~\ref{lem:RelayHyperarcs}, we
can device easy algorithms for computing the set $\mathcal{H}_{r}$.
Note, that the switch function for each relay hyperarc can be
computed in a similar manner as for the source hyperarcs. At the
same time, a particular switch function could constitute two
sub-switch functions each for two perpendicular bisectors.

%
%
%

\begin{figure}[htp]
  \begin{center}
\psfrag{s}{$s$}
   \psfrag{1}{$d_{1}$}
   \psfrag{2}{$d_{2}$}
   \psfrag{3}{$d_{3}$}
   \psfrag{c1}{$C_{1}$}
   \psfrag{c2}{$\mathcal{R}_{21}$}
\psfrag{c4}{$\mathcal{R}_{4 \infty}$} \psfrag{c3}{$\mathcal{R}_{32}$}
\psfrag{a1}{($r,d_{1},d_{2},d_{3}$)}
\psfrag{a2}{($d_{1},r,d_{2},d_{3}$)}
\psfrag{a3}{($d_{1},d_{2},r,d_{3}$)} \psfrag{e}{($C_{1},r_{2}$)}
\psfrag{f}{($C_{1},r_{1}$)} \psfrag{g}{($\mathcal{R}_{21},r_{1}$)}
\psfrag{h}{($\mathcal{R}_{21},r_{2}$)} \psfrag{i}{($\mathcal{R}_{21},r_{3}$)}
\psfrag{j}{($\mathcal{R}_{32},r_{4}$)} \psfrag{k}{($\mathcal{R}_{21},r_{4}$)}
\psfrag{l}{($C_{1},r_{4}$)} \psfrag{m}{($C_{1},r_{3}$)}
   \psfrag{a}{($d_{2},d_{1},d_{3}$)}
   \psfrag{b}{($d_{1},d_{2},d_{3}$)}
\psfrag{c}{($d_{1},d_{3},d_{2}$)} \psfrag{d}{($d_{3},d_{1},d_{2}$)}
\subfigure[Source
hypergraph.]{\label{fig:edge-a}\includegraphics[scale=0.24]{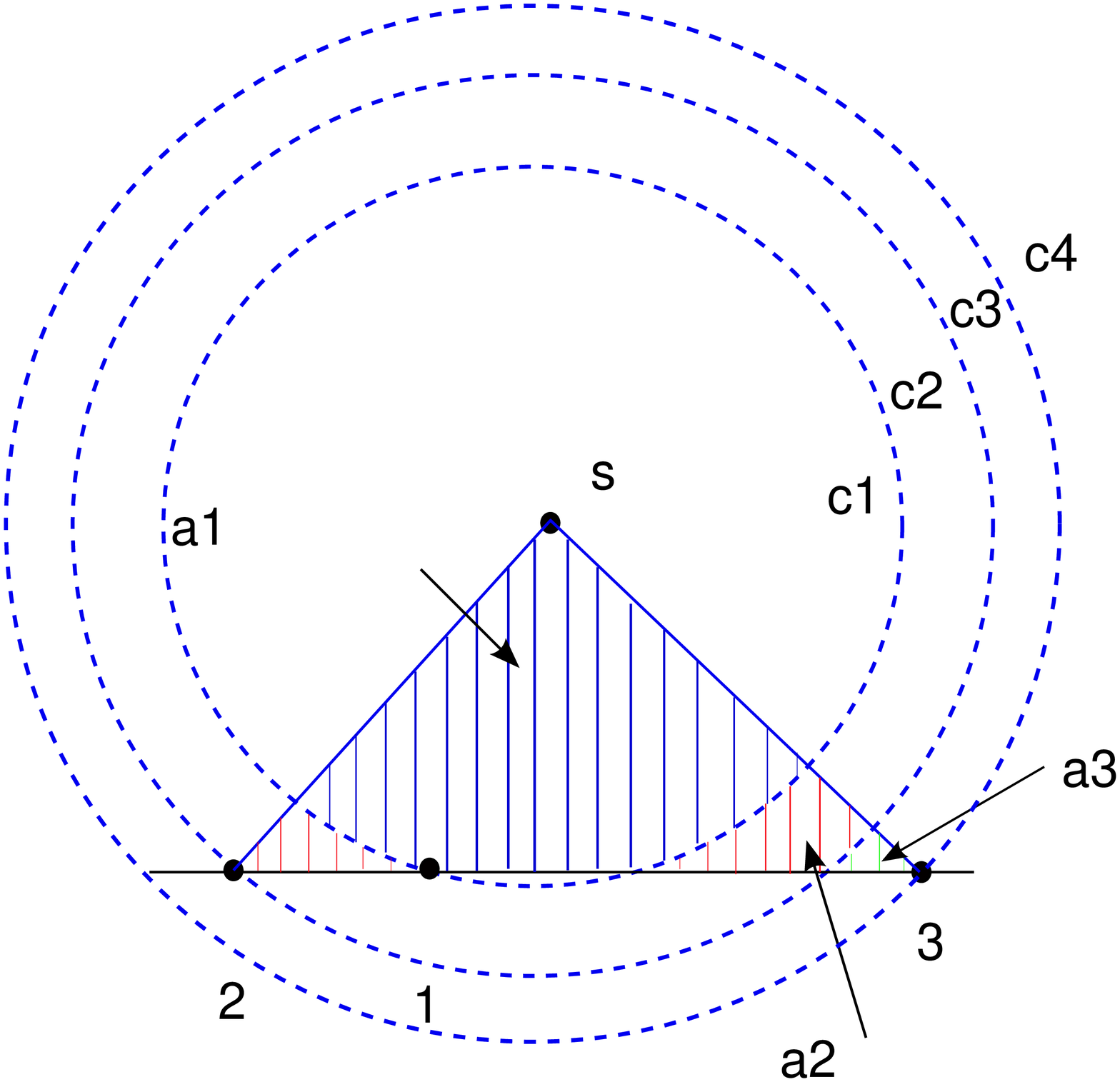}}
    \subfigure[Relay hyperarcs.]{\label{fig:edge-b}\includegraphics[scale=0.29]{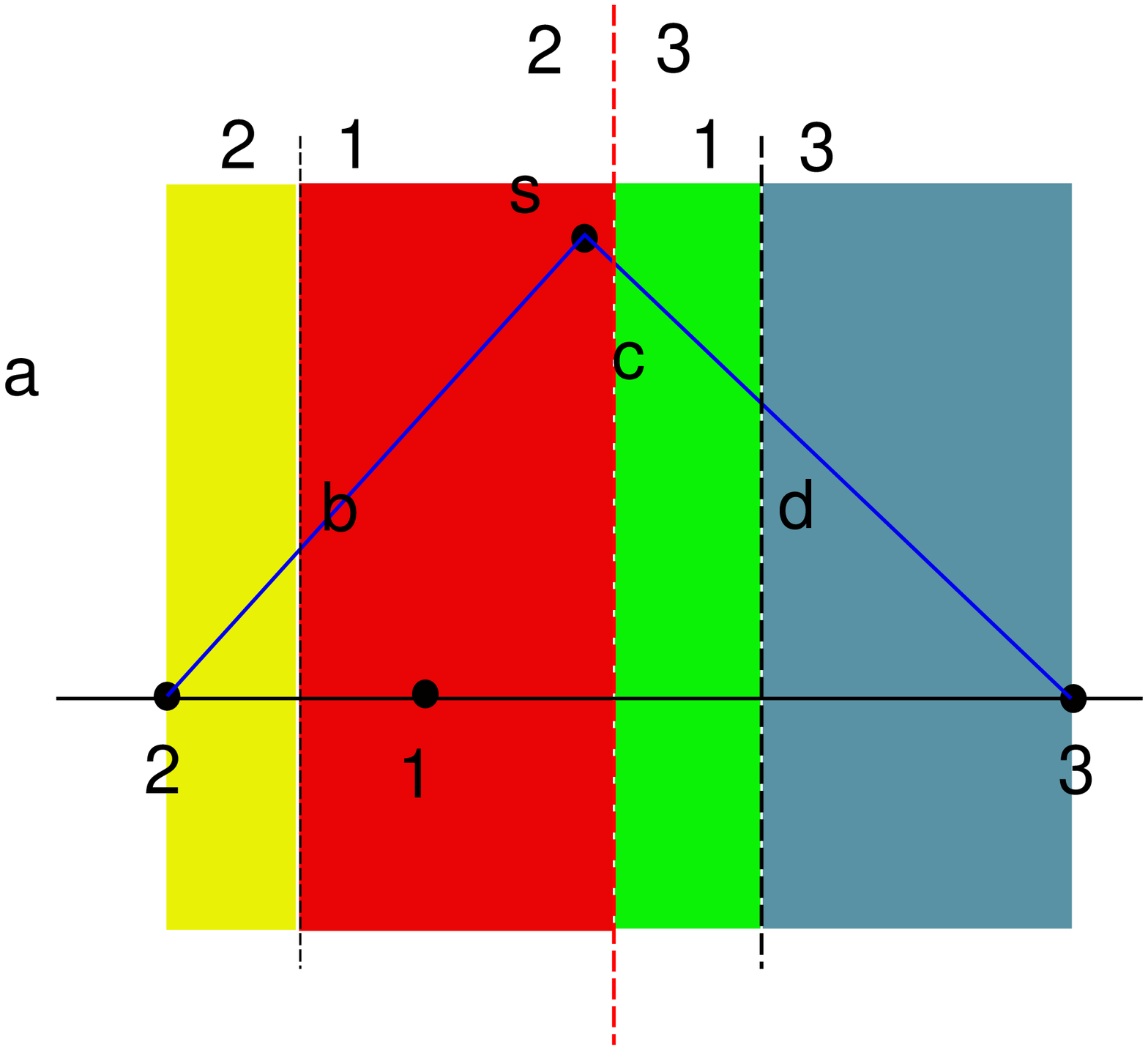}} \\
    \subfigure[Source-destination paths.]{\label{fig:edge-c}\includegraphics[scale=0.29]{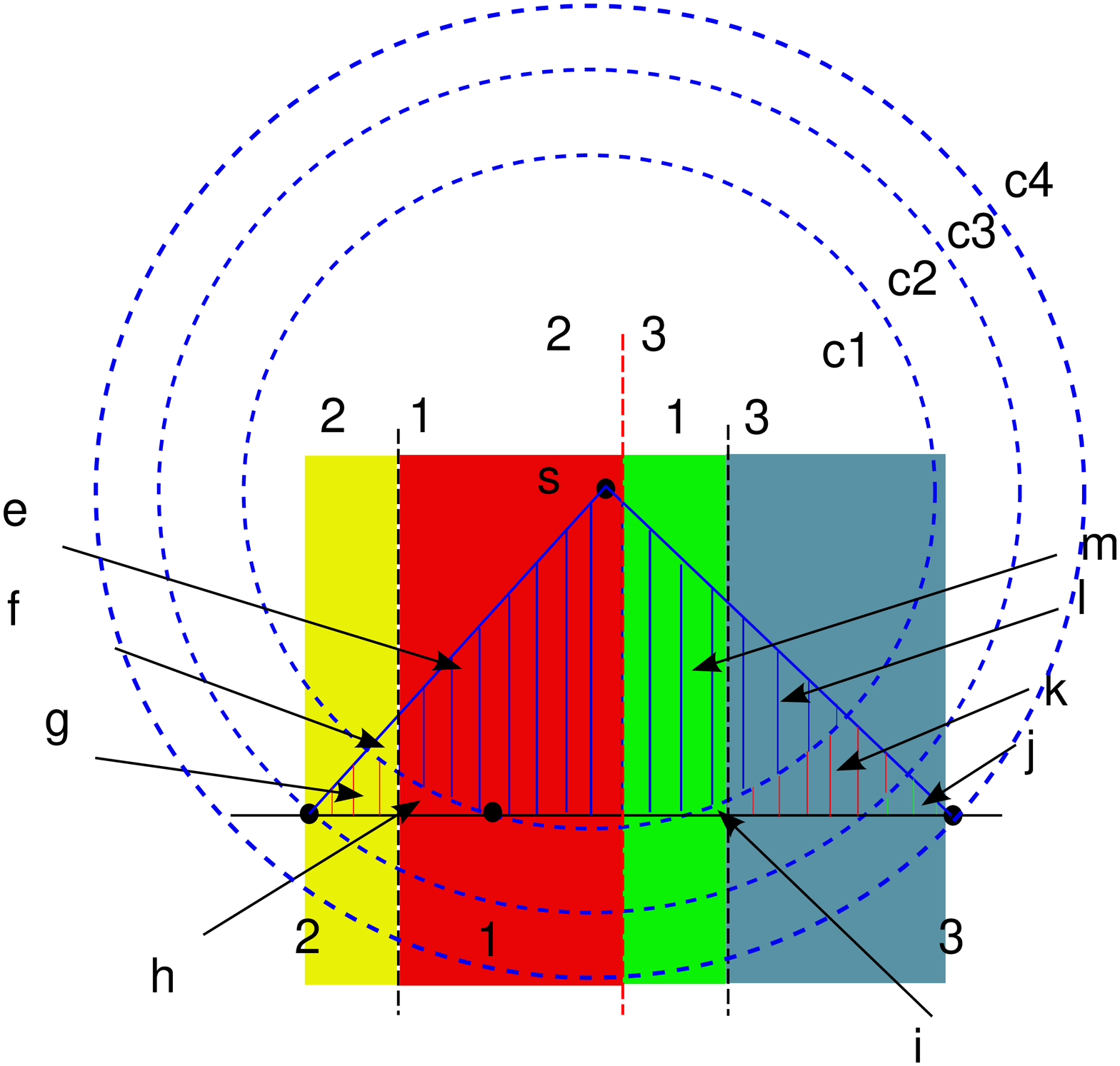}}
  \end{center}
\vspace{-2mm}
  \caption{Pre-processing. The triangle $\triangle_{sd_{2}d_{3}}$ shows $\mathcal{C}$,
  with circles and perpendicular bisectors dividing $\mathcal{C}$ is closed and disjoint sets.
  (a): Shows ordered $4$-tuple set for each region carved by the circles $c_{1},c_{2}$ and $c_{3}$ as $C_{1},\mathcal{R}_{21}, \mathcal{R}_{32}$ and $\mathcal{R}_{4 \infty}$, respectively.
  (b): Shows the ordered $3$-tuple sets of destination nodes with respect to $r$. (c): Shows the previous $2$ figures superimposed
showing the disjoint convex regions with the ordered sets of closest
nodes with respect to $s$ and $r$, respectively. The 2-tuple $(C,r)$
represents the ordered sets for each region, respectively. Here,
$r_{1},r_{2},r_{3}$ and $r_{4}$ represents the four regions in (b).}
  \label{Pre-processing.}
\vspace{-2mm}
\end{figure}

\vspace{1mm}
\subsubsection{Source-destination paths}
After successfully computing all distinct source ($\mathcal{H}_{s}$)
and relay ($\mathcal{H}_{r}$) hyperarcs for a given system, we now
need to compute all the paths from $s$ to all destinations $d_{i}
\in T$, in order to successfully cast our problem as a network
optimization program. We prefer a path based formulation as opposed
to a more basic and standard link based formulation because the path
based formulation is far well suited for convex approximations of
originally non-convex network optimization programs, in our
framework.

There are many efficient ways (polynomial time algorithms) to
compute the paths from the set $H_{s}$ and $H_{r}$. For simplicity,
we prefer here to take all combinations of the hyperarcs in $H_{s}$
and $H_{r}$. Not all these paths will be active for a certain relay
position, but the switch functions will take care of
activating/deactivating the paths. In this way, we get an upper
bound $(3n-1)(n^{2}) = |\mathcal{H}_{s}| \times |\mathcal{H}_{r}|$
on the paths from $s$ to each $d_i \in T$, out of which only a
certain number of paths will have non-zero min-cut, i.e. activated
hyperarcs. We define,  $\Omega=[1,(3n-1)(n^{2})]$ as the set of all
paths from $s$ to every $d_i \in T$. This makes the problem size
bigger, but saves cost of activated path computation. Note that the
total number of paths are polynomially bounded in our model.

\vspace{-1mm}
\subsection{Optimization.}
\vspace{-2mm} In this section we formalize the problem of optimal
relay location maximizing the multicast rate from $s$ to $T$. The
optimization constraints can be grouped into two categories: the
posynomial constraints that can be easily rewritten using
exponential transformation as convex constraints (Geometric
Programming); and the non-posynomial constraints that can only be
approximated as convex constraints. Since, all the constraints are
coupled through variables, almost all the variables and hence
constraints will go through the exponential transformation. Now, we
classify and discuss the troubling constraints in our formulation.

\subsubsection{Hyperarc rate constraints (Posynomials)}
We show an example of hyperarc rate constraints. Consider the
hyperarc $(s,d_{1}rd_{2})$ in the scenario in
Figure~\ref{fig:edge-a}, which is active when the relay is inside
the ring $\mathcal{R}_{21}$. The non-convex rate inequality can be
expressed as:\vspace{-2mm}
\begin{equation}
 y_{s,1r2} \leq \frac{P_{s1r2}}{D_{s2}^{\alpha}N_{0}} f_{s1r2}
\vspace{-1mm}
\end{equation}
where $f_{s1r2} =f^{1}_{s1r2}f^{2}_{s1r2}$, $f^{1}_{s1r2} \leq
\left(1+\gamma e^{(-\gamma z^{1}_{s1r2})} \right)^{-1}$,
$f^{2}_{s1r2} \leq \left(1+\gamma e^{(-\gamma z^{2}_{s1r2})}
\right)^{-1}$, $D_{sr}-D_{s1}=z^{1}_{s1r2}$ and
$z^{2}_{s1r2}=D_{s2}-D_{sr}$. Notice, how $f^{1}_{s1r2}$ and
$f^{2}_{s1r2}$ and their respective $z$ variables are different.
When, $f^{1}_{s1r2}$ and $f^{2}_{s1r2}$ take the value as $1$, the
hyperarc will have a non-zero min-cut. Rewriting them together,
\begin{align}
\frac{R_{s1r2}D_{s2}^{\alpha}N_{0}}{P_{s1r2}f_{s1r2}} \leq  1, \hspace{2mm} \frac{R_{s1r2}D_{s2}^{\alpha}N_{0}}{P_{s1r2}f^{1}_{s1r2}f^{2}_{s1r2}} \leq  1, \label{eq:HyperarcCons1}\\
f^{1}_{s1r2}\left(1+\gamma e^{(-\gamma z^{1}_{s1r2})} \right)  \leq  1, \label{eq:HyperarcCons2}\\
f^{2}_{s1r2}\left(1+\gamma e^{(-\gamma z^{2}_{s1r2})} \right) \leq 1, \label{eq:HyperarcCons3}\\
D_{sr}-D_{s1} \leq z^{1}_{s1r2}, \hspace{2mm} z^{2}_{s1r2} \leq
D_{s2}-D_{sr}. \label{eq:HyperarcCons4}
\end{align}
Note,  that inequalities (\ref{eq:HyperarcCons1}),
(\ref{eq:HyperarcCons2}), (\ref{eq:HyperarcCons3}) and
(\ref{eq:HyperarcCons4}) are posynomials, and $D_{s1}$ is a
constant. Using GP transformation these inequalities can be easily
converted to convex constraints. Similar argument goes for switch
functions of all other hyperarcs.

\subsubsection{Distance function constraints (Non-posynomials)} Variables $D_{uv}$ in rate
inequalities represents distance functions, given by e.g.,
\vspace{-4mm}
\begin{align}
\sqrt{(x_{r}-x_{s})^{2}+(y_{r}-y_{s})^{2}} = D_{sr}, \label{eq:DistanceCons1}
\end{align}
where, $(x_{s},y_{s})$ are fixed coordinates of $s$ and
$(x_{r},y_{r})$ are the variable coordinates of $r$. The negative
coefficients in (\ref{eq:DistanceCons1}) prohibits the use of GP
techniques. There are techniques to get around this problem
\cite{Lundell-2009}, \cite{Jung-Fa-2009}, but the extra
pre-processing cost incurred is very high in addition to the
introduction of many new variables and combinatorial constraints.

We prefer to handle the issue in a simpler manner by approximation.
Let, the only variable transformed using GP in (\ref{eq:DistanceCons1}) be $D_{sr}$. Then, we can rewrite
\begin{eqnarray}
u^{2}_{sr} + v^{2}_{sr}  \leq  e^{2D'_{sr}}, \hspace{1mm}
x_{r}-x_{s} \leq u_{sr}, \hspace{1mm}y_{r}-y_{s} \leq v_{sr}.
\label{eq:DistanceCons2}
\end{eqnarray}
The first inequality in (\ref{eq:DistanceCons2}) is non-convex.
Using the $p$-norm surrogate approximation (\cite{Li-1999}) for
(\ref{eq:DistanceCons2}), we get \vspace{-2mm}
\begin{equation}
\bigg( \frac{u^{2}_{sr} + v^{2}_{sr}}{e^{2D'_{sr}}} \bigg)^{p} \leq 1, \label{eq:DistanceCons3}
\end{equation}
where $ p\in [1,+\infty)$. Over a compact set of variables and in
the limit of $p \rightarrow \infty$, (\ref{eq:DistanceCons3})
becomes convex. In our case, for values of $p=4$ or $5$, we get good
approximation. Note, that only the first inequality in
(\ref{eq:DistanceCons2}) needs to be approximated, and since the
variables $(u_{sr},v_{sr},x_{r},y_{r})$ don't undergo GP
transformation, the rest of the inequalities in
(\ref{eq:DistanceCons2}) remain linear.

All other constraints in the program are posynomials, as we will
shortly see, so they can easily be transformed into convex
constraints. It should be noted, that it is only because of the use
of switch functions that the program becomes continuous. In
addition, carefully designing the switch function results in
posynomial hyperarc rate constraints.


\subsubsection{Network Optimization problem formulation}
Since there are $\Omega = (3n-1)(n^{2})$ number of paths for each
destination $d_{i} \in T$ only a subset of them will actually be
active (i.e. with min-cut $> 0$). $r^{d_{i}}_{q}$ as the rate on
path $q$ to destination $d_i$, where $q \in \Omega$. Recall,
$\mathcal{H}=\mathcal{H}_{s} \cup \mathcal{H}_{r}$ and the total
rate to a destination $d_i$ be defined as $r_{i} = \sum_{q \in
\Omega} r^{d_{i}}_{q}$. Also, let $v_{m} \in V$ be the farthest node
from $i$ for hyperarc $(u,V) \in \mathcal{H}$. Then, the
optimization program is,
\begin{eqnarray}
& \mbox{Program (B): } & \mbox{maximize } (R_{m}) \nonumber \\
& \mbox{subject to: } & R_{m} \leq   r_{i}, \forall i \in \N_n,\\
& & r_{i} \leq \displaystyle\sum_{q \in \Omega} r^{d_{i}}_{q}, \hspace{2mm} \forall d_{i} \in T, \\
& & \hspace{-30mm} \displaystyle\max_{((i,q)|(u,V) \in q)} R_{q}^{i} \leq \frac{P_{uV}}{D^{\alpha}_{uv_{m}}N_{0}} f_{uV}, \forall (u,V) \in \mathcal{H},\\
& & f_{uV} \leq f^{1}_{uV}f^{2}_{uV}, \hspace{2mm} \forall (u,V) \in \mathcal{H},\\
& & \hspace{-25mm}f^{l}_{uV} \leq (1+\gamma e^{-(\gamma
z^{l}_{uV})})^{-1},
\hspace{2mm} l \in [1,2],\forall (u,V) \in \mathcal{H}, \\
& & \hspace{-22mm} z^{l}_{uV} \leq D_{uv_{l}} - D_{ur}, \hspace{2mm} l \in [1,2],(u,V) \in
\mathcal{H}, \\
& & \hspace{-10mm} u_{uV}^{2} + v_{uV}^{2} \leq D^{2}_{ur}, \hspace{2mm} \forall (u,V) \in \mathcal{H}, \label{eq:pnorm1}\\
& & \hspace{-20mm} \displaystyle\sum_{((u,V) \in \mathcal{H}_{s})}
P_{uV} \leq P_{s}, \hspace{2mm} \displaystyle\sum_{((u,V) \in
\mathcal{H}_{r})} P_{uV} \leq P_{r},
\end{eqnarray}
\begin{eqnarray}
& \mbox{where,} & {x}_{r} - x_{u} \leq u_{uV}, \hspace{2mm} \forall (u,V) \in \mathcal{H},\\
& & {y}_{r} - y_{u} \leq v_{uV},\hspace{2mm} \forall (u,V) \in \mathcal{H},\\
& & \hspace{-20mm} x_{r} \geq 0, y_{r} \geq 0, {y}_{r} \leq \lambda
{x}_{r}, \hspace{1mm} {y}_{r} + \lambda' {x}_{r} \leq \eta.
\end{eqnarray}
In the above program $(x_{r},y_{r})$ are variable relay coordinates
and (27) captures constraints that make $\mathcal{C}$. Program (B)
is a non-convex program expressed in posynomial and polynomial
inequalities. Applying GP transformation to the following variables
$\{r^{d_i}_{p},P_{uV},D_{uV},f_{uV}\}$, $p$-norm approximation to
constraints (\ref{eq:pnorm1}) and leaving the rest of the variables
unchanged, we get the following convex approximation,
 \vspace{-1mm}
\begin{eqnarray}
\hspace{-20mm} \mbox{Program (C): }  \mbox{maximize } \min_{d_{i} \in T} \left(\sum_{q \in \Omega}{r'^{d_i}_{q}} \right) \nonumber \\
 \mbox{subject to: } \begin{split}
 N_{0} & e^{(r'^{i}_{uV} + \alpha D'_{uv_{m}} - P'_{uV} -f'_{uV})} \leq 1, \\
  & \forall q \in (u,V), \forall (u,V) \in \mathcal{H},
\end{split}\\
 e^{(f'_{uV}-f'^{1}_{uV}-f'^{2}_{uV})} \leq 1, \hspace{2mm} \forall
 (u,V) \in \mathcal{H},\\
 e^{(f'^{l}_{uV})} + e^{(f'^{l}_{uV} - \gamma z^{l}_{uV})} \leq 1,
 \hspace{2mm} l \in [1,2], \forall (u,V) \in \mathcal{H}, \\
 z^{l}_{iJ} +e^{D'_{ur}} \leq D_{uv}, \hspace{4mm} \forall (u,V) \in
 \mathcal{H},\\
\left( \frac{u_{uV}^2+v_{uV}^2}{e^{D'_{ir}}}\right)^{p} \leq 1,
\hspace{2mm} \forall
 (u,V) \in \mathcal{H}, \label{eq:OptimConsC}\\
 \displaystyle\sum_{((u,V) \in \mathcal{H}_{s})}
e^{P'_{uV}} \leq P_{s}, \hspace{2mm} \displaystyle\sum_{((u,V) \in
\mathcal{H}_{s})} e^{P'_{uV}} \leq P_{r},
\end{eqnarray}
where, we have used the GP transformation $x'=\log(x)$ ($x$ is the
original variable of program (B)) on certain variables.

Program (C) is a convex approximation of program (B) with no
underlying combinatorial hard structure. The approximation is only
coming from constraint type (\ref{eq:OptimConsC}) using $p$-norm
surrogation technique that gives a convex approximation to the
constraint (\ref{eq:pnorm1}) in program (B).

Note, that the objective function is modified in (C), instead of
having the sum of positive exponential terms, we have replaced it
 by a sum of linear functions, which is far easier to maximize. The maximizers of the
function $\displaystyle\max_{X} (x_{1}+x_{2})$ also maximizes the
function $\displaystyle\max_{X}
 (e^{x_{1}}+e^{x_{2}})$, over a compact set $X$ with certain particular characteristics.
 This is generally not true, but in our case
 due to network coding constraints with certain simple tricks it can be proven that it holds true. Due to the
 lack of space we omit the detailed proof.

As we know, with the increasing value of $p$, the program approaches
a complete convexity with zero duality gap, thus standard convex
optimization algorithms could be used to solve problem (C) with
increasing accuracy.
The optimal values of program (B)
could be easily constructed from the optimal values of program (C).
\subsection{Simulations}
\begin{figure}[tp]
\begin{center}
\includegraphics[width=0.9\columnwidth]{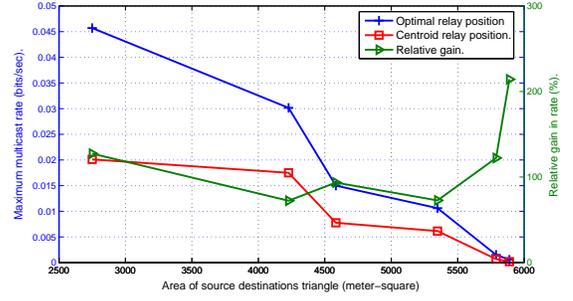}
\end{center}
\vspace{-4mm}
\caption{{BRC with $n=2$ destinations. Rates for $r$ located at the centroid of $\triangle_{s d_1 d_2}$, at the optimal position, and relative gain.}}
\label{fig:SimuResults}
\vspace{-6mm}
\end{figure}

In this section, we present simulation results for the BRC with $n =
2$ destinations. We compare the multicast rate obtained by
optimizing the relay location and the source in addition to relay
power allocations, with the case where the relay is located at a
naive yet seemingly interesting position: the centroid of triangle
$\triangle_{s d_1 d_2}$ , and only the power allocations are
optimized. The simulations are run for an increasing size of the
area of $\triangle_{s d_1 d_2}$ and a random network topology for
each area is chosen.

Figure~\ref{fig:SimuResults} shows the maximum multicast rate (blue
and red) for optimal and centroid relay positions respectively. The
SNR $\frac{P}{N_0}$ is normalized to 1. Note that the actual values
of the rates are not as important because of the normalization, for
higher power values, the rate would certainly have higher values.
For increasing area of triangle $\triangle_{s d_1 d_2}$, the maximum
multicast rate tends to drop, which is due to the constrained power
and larger distances, but the relative gain goes up. This implies
that for farther placed nodes the sensitivity of the relay location
is higher and can produce significant gains in rate for the optimal
relay location. Its clear from the results in
Figure~\ref{fig:SimuResults} that the centroid is not the optimal
location. The rise in the relative gain becomes very strong due to
the fact that the low-SNR regime is more sensitive to the location
of nodes (hence, distances) that determine the hyperarc rates in the
limit of disappearing SNR as opposed to e.g. in high SNR regime,
where a displacement of $\pm \epsilon$ for the location of $r$ would
not effect the rate significantly.

\section{ARBITRARY CASE} \label{sec:Arbitrary}
\vspace{-2mm} In this section, we answer the same set of questions
but for arbitrarily placed source and destination nodes. Almost all
concepts can be carried over to, straightforwardly.


The steps of the pre-processing stage can be summarized as

\noindent Input: $\{s, T \}$ set of nodes with their cartesian
coordinates.
\begin{enumerate}
\item Compute convex hull $\mathcal{C}$.
\item Compute $\mathcal{H}_{s}$, switch functions (using Lemma 2).
\item Compute the disjoint convex regions by
superimposing all $k$-order Voronoi diagrams of $\mathcal{C}$.
\item Compute $\mathcal{H}_{r}$ (using relay hyperarc algorithm).
\end{enumerate}
Output: $\mathcal{H} =\mathcal{H}_{s} \cup \mathcal{H}_{s}$.

Once we have $\mathcal{H}_{s}$ and $\mathcal{H}_{s}$, we can compute
$\Omega=[1,|\mathcal{H}_{s}| \times |\mathcal{H}_{r}|]$. Ultimately,
the optimization program could be stated as, \vspace{-6mm}
\begin{eqnarray}
&\mbox{Program (D): }& \mbox{maximize } \left( g_{0}(x) \right)
\nonumber
\end{eqnarray}
\vspace{-6mm}
\begin{eqnarray}
& \mbox{subject to:} & g_{i}(x) \leq 1, \hspace{2mm} i \in [1,k], \label{eq:progD1}\\
& & g_{j}(x) \leq 1, \hspace{2mm} j \in [k+1,K], \label{eq:progD2}
\end{eqnarray}
where, the constraints (\ref{eq:progD1}) are the posynomials
constraints that can be transformed to convex convex constraints
using GP and constraints (\ref{eq:progD2}) are the non-posynomial
constraints that are approximated using $p$-norm approximation. The
objective function $g_{0}(x)$ represents the multicast rate. In
program (D), $x$ is a vector of variables.

Program (D), is an abstract representation of the actual program.
Since, the program (D) is simply program (B) (but for arbitrary
placement of destination nodes), the structure of (D) is the same as
(B). The main difference is the pre-processing stage for the two
cases, in this case which involves computation of $k$ nearest
neighbor nodes and superimposing them to form disjoint regions in
$\mathcal{C}$ for distinct $n$-nearest neighbors. We would like to
note, that this computation, although polynomially bounded, can be
heavy. There are many polynomial time algorithms in the literature
of computational geometry that solve this problem efficiently,
\cite{Edelsbrunner-book1987}.

\section{RESULTS AND CONCLUSION} \label{sec:Conclusion}
\vspace{-2mm} A comprehensive and efficient solution is developed to
model and answer the problem of optimal relay positioning so as to
maximize the multicast rate from the source $s$ to the destination
set $T$ in a low-SNR network. The proposed solution is a non-convex
network optimization problem in its basic form that is difficult to
solve. Using GP, switch functions and $p$-norm surrogate
approximation we transform the problem to a convex approximation
that can be solved using standard convex optimization algorithms.

 To abridge, the important
contributions of this work could be summed up in the following
words: we present a comprehensive approach to determine the optimal
relay position under the pretext of network optimization problem.
Network topologies consisting single source, multiple destinations
with the only intermediate node as relay are considered in complete
generality on a $2$-D plane. Using superposition coding and
frequency division we construct a wireline like hypergraph. The
low-SNR hyperarc model using superposition coding provides an
interference free network model that is easily scalable to complex
network topologies. Using the tools of computational geometry and
network optimization, we presented a network optimization framework
based solution that is intuitive and easy to understand. Also, we
show that positioning the relay optimally significantly affect the
network performance.

In addition, the main causes for complexity in our approach are the
pre-processing stage and the non-convexity arising from
non-posynomial constraints upon GP transformation. The former reason
could be somewhat tolerated, as generally for solving network
planning problems heavy pre-processing is required. In contrast, in
our case the pre-processing stage consists of polynomial time
operations at the cost of only slight sub-optimality in
approximation.

The questions our work answers are just a fraction of the
interesting questions that it opens up. An interesting direction
would be to extend this model to general multicommodity flow
optimization problems involving more number of relay nodes. On the
other hand, from the computational point of view, an interesting
question is how we can efficiently  build the exact number of paths
essentially bringing down the size of the network optimization
program. Finding other techniques to model this problem could be
interesting, e.g. utilizing geometric properties of the problem and
ways to bring down the computational complexity.

\bibliographystyle{./Biblio/IEEEtran}

\bibliography{./IEEEabrv,./bibLowSNR,./bibNC,./bibCellularStd,./bibOptim}

\end{document}